\documentclass[sigconf,authorversion,nonacm]{acmart}

\usepackage{cite}
\usepackage{bussproofs}
\usepackage{stmaryrd}
\usepackage{cmll}
\usepackage{mathrsfs}
\usepackage{algorithmic}
\usepackage{graphicx}
\usepackage{textcomp}
\usepackage{xcolor}
\usepackage{proof}
\usepackage{cleveref}

\newtheorem{theorem}{Theorem}
\newtheorem{proposition}[theorem]{Proposition}
\newtheorem{lemma}[theorem]{Lemma}

\theoremstyle{definition}
\newtheorem{definition}[theorem]{Definition}

\newcommand{\Ex}{\mathrm{Ex}}
\newcommand{\meas}[1]{\mathopen{\llbracket}#1\mathclose{\rrbracket}}

\newcommand{\A}{\mathscr{A}}

\DeclareMathOperator*{\bigmeet}{\bigcurlywedge}
\DeclareMathOperator*{\bigjoin}{\bigcurlyvee}
\DeclareMathOperator{\preq}{\preccurlyeq}

\newcommand{\varset}{\mathcal{V}}

\DeclareMathOperator{\li}{\multimap}

\DeclareMathOperator{\tens}{\otimes}

\newcommand{\orth}[1]{#1^\bot}
\newcommand{\1}{\mathbf{1}}

\DeclareMathOperator{\dom}{dom}
\DeclareMathOperator{\lsep}{Sep_L}

\DeclareMathOperator{\lcore}{S^0_L}
\DeclareMathOperator{\ecore}{S^0_!}
\DeclareMathOperator{\elemcore}{S^0_\sharp}
\DeclareMathOperator{\uacl}{\uparrow @}
\DeclareMathOperator{\equiprov}{\dashv\vdash}

\newcommand{\I}{\mathbf{I}}
\newcommand{\B}{\mathbf{B}}
\newcommand{\C}{\mathbf{C}}
\newcommand{\K}{\mathbf{K}}
\newcommand{\W}{\mathbf{W}}
\newcommand{\cS}{\mathbf{S}}

\newcommand{\cmb}{\mathbf{A}}

\newcommand{\Ke}{\mathbf{K_!}}
\newcommand{\We}{\mathbf{W_!}}
\newcommand{\D}{\mathbf{D}}
\newcommand{\cd}{\mathbf{\delta}}
\newcommand{\F}{\mathbf{F}}

\newcommand{\perm}{\mathfrak{S}}

\DeclareMathOperator{\pow}{\mathcal{P}}

\DeclareMathOperator{\ex}{Ex}
\newcommand{\m}[2]{\llbracket #1, #2 \rrbracket_m}
\newcommand{\proj}{\mathfrak}
\newcommand{\Proj}{\mathfrak{P}}
\DeclareMathOperator{\antipode}{\mathbf{\aleph}}
\newcommand{\sorth}[1]{#1^\simperp}
\newcommand{\types}{\mathbf{\Pi}}
\newcommand{\pid}{\proj{id}}

\DeclareMathOperator{\red}{\twoheadrightarrow}

\DeclareMathOperator{\bred}{\red_\beta}

\DeclareMathOperator{\fv}{FV}

\DeclareMathOperator{\bangred}{\red_!}

\DeclareMathOperator{\enteq}{\dashv\vdash}
\DeclareMathOperator{\ent}{\vdash}
\newcommand{\cl}[1]{[#1]}

\newcommand{\RecordFont}[1]{\texttt{#1}}
\newcommand{\Recl}{\RecordFont{l}}
\newcommand{\Appl}[1]{#1\RecordFont{.l}}
\newcommand{\Recr}{\RecordFont{r}}
\newcommand{\Appr}[1]{#1\RecordFont{.r}}
\newcommand{\Record}[1]{\left\{ #1 \right\}}

\newtheorem{remark}[theorem]{Remark}
\newtheorem{example}[theorem]{Example}
\newtheorem*{theorem*}{Theorem}
\newtheorem*{corollary*}{Corollary}

\theoremstyle{definition}
\newtheorem*{definition*}{Definition}

\begin{document}

\title{Linear Realisability and Implicative Algebras}

\author{Alexandre Lucquin}
\email{lucquin@lipn.univ-paris13.fr}
\affiliation{%
  \institution{Université Paris Nord}
  \city{Villetaneuse}
  \country{France}
}

\author{Luc Pellissier}
\email{luc.pellissier@u-pec.fr}
\orcid{0000-0003-1923-8193}
\affiliation{%
  \institution{Université Paris-Est Créteil}
  \city{Créteil}
  \country{France}
}

\author{Thomas Seiller}
\email{thomas.seiller@cnrs.fr}
\orcid{0000-0001-6313-0898}
\affiliation{%
  \institution{CNRS}
  \city{Paris}
  \country{France}
}

\begin{abstract}
Realizability, introduced by Kleene, can be understood as a concretization of
the Brouwer-Heyting-Kolmogorov (BHK) interpretation of proofs,
providing a framework to interpret mathematical statements and
proofs in terms of their constructive or computational content. Over
time, this concept has evolved through various extensions, such as
Kreisel’s modified realizability or Krivine's classical realizability. Parallel to these developments, Girard’s work on linear logic introduced another perspective, often seen as another concrete realization of the BHK interpretation. The resulting constructions, encompassing
models like geometry of interaction, ludics, and interaction
graphs, were recently unified under the term linear realizability
models to stress the intuitive connection with intuitionnistic and classical realizability. 

The present work establishes for the first time a formal link between linear realizability models and the realizability constructions
of Kleene and Krivine. Our approach leverages Miquel’s framework:
just as linear logic can be viewed as a decomposition of intuitionistic
and classical logic, we propose a linear decomposition of implicative algebras and show that linear realisability models provide concrete examples of such decompositions.
\end{abstract}

\maketitle

\section{Introduction}


Introduced by Kleene in 1945 \citep{kleene} realizability was developed to analyze constructive proofs and derive supplementary insights from them. Realizability can also be understood as a concretization of the Brouwer-Heyting-Kolmogorov (BHK) interpretation of proofs, providing a framework to interpret mathematical statements and proofs in terms of their constructive or computational content. Over time, this concept has evolved through various extensions, such as Kreisel’s modified realizability \citep{kreisel}, all of which operate within the domain of intuitionistic logic. This restriction is natural, given the constructive foundations of realizability. While Kleene's \emph{number realizability} was build upon the recursive functions, more recent work have been considering the lambda-calculus as the underlying computational model, notably to establish bridges between computational and logical principles, such as bar induction and the axiom of countable choice \citep{BerardiCoquand}.

Realizability topoi \citep{vanOosten} form a bridge between realizability theory and the general framework of topos theory, encapsulating computational phenomena within a categorical structure. They arise from partial combinatory algebras (PCA), which may be regarded as providing a notion of untyped computation. A topos can then be constructed via the tripos-to-topos construction, or via assemblies. A tripos is a structure that encodes logical information, and in the case of realizability, it captures the logical principles inherent in the PCA. The resulting realizability topos is a generalization of the category of sets, enriched with an internal logic reflective of the computability encoded in the PCA. These topoi are important instances within the broader class of topoi, connecting specific computational interpretations to the abstract framework of topos theory, and demonstrating how logic, computation, and category theory interplay at a foundational level.

In recent years, Krivine introduced a groundbreaking development with classical realizability \citep{krivine1,krivine2,MiquelReal}. By considering the lambda calculus extended with the \texttt{call-cc} operator, which encapsulates the computational essence of classical principles \citep{griffin}, Krivine developed realizability constructions that extend to classical logic. This approach moreover realizes the axioms of Zermelo-Fraenkel set theory (excluding the axiom of choice) and leads to the construction of novel models of set theory \citep{krivinenewmodels}.

While a connection between traditional realizability constructions and Krivine’s classical realizability was anticipated \citep{streicher}, it was Miquel’s work that formally unified these frameworks. Miquel axiomatized realizability constructions in a manner that can systematically derive models for both intuitionistic and classical logic, bridging these two domains. The framework, \emph{implicative algebras}, also allows to properly account for the categorical aspects of realisability: each implicative algebra naturally gives rise to a tripos \citep{tripostheory}.

Parallel to these developments, Girard’s work on linear logic introduced another perspective, often seen as another concrete realization of the BHK interpretation. Girard's constructions—spanning models like geometry of interaction \citep{goi1,goi2,goi3,goi5}, ludics \citep{locussolum}, and interaction graphs \citep{seiller-goim,seiller-goiadd,seiller-goig,seiller-goie,seiller-goif,seiller-markov}—were recently unified under the term \emph{linear realizability} models by Seiller \citep{seiller-hdr}. This nomenclature underscores the striking similarities between Girard’s and Krivine’s techniques, suggesting potential connections between linear realizability models and both Kleene or Krivine’s approaches. However, these connections have remained elusive due to critical, albeit subtle, mismatches in their constructions.


\paragraph{Contributions}
The present work establishes for the first time a formal link between Girard’s linear realizability models and the realizability constructions of Kleene and Krivine. Our approach leverages Miquel’s framework: just as linear logic can be viewed as a decomposition of intuitionistic and classical logic, we propose a linear decomposition of implicative algebras. We demonstrate that linear realizability models from the literature fit naturally into this framework. We then explore the relationship with Miquel’s implicative algebras and the topos construction.

By establishing that linear realizability models give rise to implicative algebras, this work opens numerous avenues for further research. First, it is crucial to explore how these models relate to existing constructions. Notably, we anticipate that linear realizability models will yield novel examples of classical realizability, potentially enriching its applications in set theory. Second, we believe our results point towards a decomposition of the resulting tripos in alignment with the linear logic decomposition of implication. This, in turn, may pave the way for the formulation of a \emph{linear tripos}, a concept that has remained elusive in the field although some progress have been made \citep{shulman}.

\section{Background}

\subsection{Implicative algebras}

Implicative algebras are an algebraic structure introduced by Miquel \citep{DBLP:journals/mscs/Miquel20} which generalizes forcing and realizability (both intuitionistic and classical), and allows us to factorize the corresponding model-theoretic constructions.

This structure has the property that the operations of the $\lambda$-calculus can be lifted into it, allowing its elements to be seen both as truth values and as (generalized) realizers.

\begin{definition}
    An implicative structure $(\A,\preq, \to)$ is a complete lattice \footnote{In \citep{DBLP:journals/mscs/Miquel20} Miquel considers equivalently complete meet semi-lattices, and says that the join only exists "by accident". Since we will use the join in some constructions, we will here speak of complete lattices} $(\A,\preq)$ equipped with a binary operation $(a,b) \mapsto (a \to b)$ called implication satisfying for all $a,a',b,b' \in \A$ :
    
    \begin{enumerate}\label{im1}
        \item if $a' \preq a$ and $b \preq b'$ then $(a \to b) \preq (a' \to b')$
        \item $a \to \bigmeet_{b \in B} b = \bigmeet_{b \in B} (a \to b)$ 
    \end{enumerate}
\end{definition}

A first intuition is that $\A$ represents a semantic type system, where $\preq$ is the relation of subtyping, and $\to$ represents the arrow type construction. We also want to think that there is some notion of realizability for which $\A$ is the set of truth values, i.e. there is a set of realizers $P$ with some closure properties, and $\A \subseteq \pow(P)$. With this intuition, we can associate to each realizer $t$ its principal type $[t]$, the smallest truth value realized by $t$. Then we can use any realizer as truth values, with the relation $t \in a$ becoming $[t] \preq a$. 

An important feature of implicative structures is that we can lift the application and lambda abstraction to the level of truth values, allowing us to use truth values as if they were realizers. With this property, we can view the elements of $\A$ as a generalized realizer where each element realizes itself and is its own principal type. In this third view, the relation $a \preq b$ can be read $a$ is a subtype of $b$, $a$ realize $b$ or $a$ is stronger than $b$ (in the sense that $a$ realize every type realized by $b$) depending on whether we choose to see $a$ and $b$ as tow truth values, a realizer and a truth value, or two realizers. With this approach, each truth value is realized, at least by itself. So we need to equip our implicative structure with a criterion of truth encompassed in what Miquel call a separator. 

\begin{definition}
Let $(A,\leq,\rightarrow)$ be an implicative structure.
A \emph{separator} is a subset $\mathcal{S} \subseteq A$ such that:
\begin{enumerate}
  \item $\mathcal{S}$ is \emph{upward closed}: if $a \in \mathcal{S}$ and
        $a \leq b$, then $b \in \mathcal{S}$.
  \item $\mathcal{S}$ is closed under \emph{modus ponens}: if
        $a \in \mathcal{S}$ and $a \rightarrow b \in \mathcal{S}$, then
        $b \in \mathcal{S}$.
\end{enumerate}
\end{definition}

From a logical point of view, a separator plays the role of a deductively closed theory; and the smallest separator, called the \emph{core} of the implicative structure, corresponds to the set of tautologies. A very important point for us is that the deduction rules are represented by the presence of all closed $\lambda$-terms in this core, so restricting the content of the core corresponds to omitting logical rules. We will use this mechanism to allow our separators to avoid structural rules, enabling them to account for linearity.

As it is, this framework is not compatible with implicative algebras, so we will introduce linear implicative algebras, a relaxed version of those which allow the omission of control rules. This new structure encompass linear realizability models (with some minor assumption, such as the existence of identity). We will show that linear implicative algebras represent intuitionnistic multiplicative linear logic, and then extend to bigger fragments of linear logic.

\subsection{Applicative structure}

In order to relate the implicative algebras approach to linear realisability, we introduce an alternative definition of implicative algebra focused on \emph{application} instead of \emph{implication}. This will be useful to show the connection with linear realizability models from the literature, since the latter are defined with a focus on application.

\begin{definition}\label{applicative structure}
    An \emph{applicative structure} $(\A,\preq, \cdot)$ is a complete lattice $(\A,\preq)$ equipped with a binary operation $(a,b) \mapsto (a \cdot b)$ called application satisfying the following axioms for all $a,a',b,b' \in \A$:
    \begin{enumerate}
        \item if $a \preq a'$ and $b \preq b'$, then $a \cdot b \preq a' \cdot b'$;
        \item $\bigjoin_{a \in A} (a \cdot b) = (\bigjoin_{a \in A} a) \cdot b$.
    \end{enumerate}
    From this, we define 
    \[
        a \leadsto b := \bigjoin \{c \in \A \mid c \cdot a \preq b\}
    \]
\end{definition}

%

\begin{proposition}\label{prop:applicative}
    Let $(\A,\preq, \cdot)$ be an applicative structure. For all $a,a',b,b' \in \A$ :
    \begin{enumerate}
        \item if $a' \preq a$ and $b \preq b'$, then $a \leadsto b \preq a' \leadsto b'$;
        \item $a \preq (b \leadsto a \cdot b)$;
        \item $(a \leadsto b) \cdot a \preq b$;
        \item $a \leadsto b = \max\{c \in \A : c \cdot a \preq b\}$;
        \item $a \cdot b \preq c$ iff $a \preq (b \leadsto c)$.
    \end{enumerate}
\end{proposition}

We will now show that applicative structures define implicative structures, and vice versa.

\begin{proposition}\label{prop:applicativetoimplicative}
     If $(\A,\preq, \cdot)$ is an applicative structure with $\leadsto$ defined as in \ref{applicative structure}, then $(\A, \preq, \leadsto)$ is an implicative structure and its application $ab$ is equal to  $a \cdot b$, for all $a,b \in \A$.
\end{proposition}

\begin{proposition}\label{prop:implicativetoapplicative}
    If $(\A, \preq, \to)$ is an implicative structure then $(\A,\preq, \cdot)$, where $a \cdot b := ab$, is an applicative structure, and $\leadsto$ as defined in \ref{applicative structure} is equal to $\to$.
\end{proposition}

    Any implicative structure $\A = (\A, \preq, \to)$ naturally induces a semantic type system whose types are the elements of $\A$. 
    
    In this framework, a typing context is a finite (unordered) list $\Gamma = x_1 : a_1, \ldots, x_n : a_n$, where $x_1, \ldots, x_n$ are pairwise distinct $\lambda$-variables and where $a_1, \ldots, a_n \in \A$. Thinking of the elements of $\A$ as realizers rather than as types, we may also view every typing context $\Gamma = x_1 : a_1, \ldots, x_n : a_n$ as the substitution $\Gamma = x_1 := a_1, \ldots, x_n := a_n$.
    
    Given a typing context $\Gamma = x_1 : a_1, \ldots, x_n : a_n$, we write $\dom(\Gamma) = {x_1, \ldots, x_n}$ its domain, and the concatenation $\Gamma, \Gamma'$ of two typing contexts $\Gamma$ and $\Gamma'$ is defined as expected, provided $\dom(\Gamma) \cap \dom(\Gamma') = \emptyset$. 
    
    Given two typing contexts $\Gamma$ and $\Gamma'$, we write $\Gamma' \preq \Gamma$ when for every declaration $(x : a) \in \Gamma$ there is $(x : b) \in \Gamma'$ such that $b \preq a$.
    
\begin{definition}[Semantic type system]
    Given a typing context $\Gamma$, a $\lambda$-term $t$ with parameters in $\A$ and an element $a \in \A$,
    we define the (semantic) typing judgment $\Gamma \vdash t : a$ a as the following shorthand:
    $$
        \Gamma \vdash t : a \quad :\iff \quad \fv(t) \subseteq \dom(\Gamma)\ \mathrm{and}\ (t[\Gamma])^\A \preq a
    $$
    (using $\Gamma$ as a substitution in the right-hand side inequality). 
\end{definition}

This semantic type system gives us some semantic typing rules that we will use to construct typing derivation from proof trees.

\begin{proposition}[Semantic typing rules (2.23)]\label{prop:semantictyping}
    For all typing context $\Gamma, \Gamma'$, $\lambda$-terms $t,u$ with parameters in $\A$, $a,a',b,b',c \in \A$ and $(a_i)_{i \in I} \in \A^I$, the following 'semantic typing rules' are valid :
    \begin{itemize}
        \item $x : a \vdash x : a$ (Axiom)
        \item $\vdash a : a$ (Parameter)
        \item If $\Gamma \vdash t : a$ and $a \preq a'$ then $\Gamma \vdash t : a'$ (Subsumption)
        \item If $\Gamma \vdash t : a$ and $\Gamma' \preq \Gamma$ then $\Gamma' \vdash t : a$ (Context subsumption)
        \item If $\Gamma, x : a \vdash t : b$ then $\Gamma \vdash \lambda x.t : a \li b$ ($\li$-R)
        \item If $\Gamma \vdash t : a$ and $\Delta, x : a \vdash u : b$ then $\Gamma, \Delta \vdash (\lambda x.u)t : b$ (Cut)
        \item If $\Gamma \vdash t : a$ and $\Delta, x : b \vdash u : c$ then $\Gamma,\Delta, y : a \li b \vdash (\lambda x.u)(yt) : c$ ($\li$-L)
        \item If $\Gamma \vdash t : a_i$ for all $i \in I$ then $\Gamma \vdash t : \bigmeet_{i \in I} a_i$ (Generalization)
    \end{itemize}
\end{proposition}

\begin{remark}
    We do not use the exact same rules as in \citep{DBLP:journals/mscs/Miquel20} but the systems are equivalent.
\end{remark}


\subsection{Linear realisability models}

Soon after the introduction of linear logic \citep{ll}, Jean-Yves Girard introduced the \emph{geometry of interaction program} \citep{towards}. Motivated by the idea of having a dynamic representation of proofs, many of the constructions introduced as part of this program since its inception also include a reconstruction of types based on an underlying dynamic situation. These constructions have been recently called \emph{linear realisability} by Seiller \citep{seiller-hdr} who provided a more abstract presentation akin to the PCA-based view on standard (intutionnistic) realizability. 

More precisely, the models are defined from a computational model, together with a measurement allowing to define types called \emph{linear realisability situation}. The definition involves the choice of a commutative group $(\Theta,+,0)$. For our purposes, this group can be considered to be the real numbers together with the usual addition.

\begin{definition}
A \emph{linear realisability situation} is a triple $(P,\Ex,\meas{\cdot,\cdot})$, where:
\begin{itemize}
\item $P$ is a set (of \emph{programs});
\item $\Ex: P\times P \rightarrow P$ is an associative operation representing the composition of programs;
\item $\meas{\cdot,\cdot}: P\times P \rightarrow \Theta$ satisfies the so-called \emph{trefoil}, or \emph{2-cocycle}, property with respect to $\Ex$:
\begin{equation} \forall p,q,r \in P, ~ \meas{\Ex(p,q),r}+\meas{p,q} = \meas{p,\Ex(q,r)}+\meas{q,r}. \label{trefoil} \end{equation}
\end{itemize}
\end{definition}

We will now sketch how models of linear logic can be constructed from a linear realisability situation. More detail on the construction and these instances can be found in Seiller's habilitation thesis \citep{seiller-hdr}. 

The construction define \emph{types} as bi-orthogonally closed set, for a notion of orthogonality induced by the measurement. This generalizes formal concepts \citep{GanterWille1999,seiller-Weyl}, and follows a technique which has been used to define denotational models of linear logic such as coherent spaces \citep{proofsandtypes,qcs,probcoh} or finiteness spaces \citep{finitenessspaces}. In the latter constructions, abstracted as categorical double glueing constructions \citep{doubleglueing}, the orthogonality satisfies a Jacobi identity $a\otimes b \perp c \Leftrightarrow a\perp b\otimes c$, allowing to lift the $\otimes$ construction to a monoidal product on types. However, in the present case, Equation \ref{trefoil} exhibits a mismatch between the associativity of execution and the measurement. The construction we now detail allows to lift the execution on types in a way that ensures associativity of the operation on types despite this mismatch. This is done by adjoining to the program an element of $\Theta$ used to twist the execution, constructing types on the set $\Theta\times P$ instead of $P$ directly.

\begin{definition}
A project is a pair $(\alpha, a)\in \Theta\times P$, written $\alpha\cdot a$. Execution and measurement are extended to projects as follows:
\begin{align*}
\Ex(\alpha \cdot a, \beta \cdot b) &= (\alpha+\beta+\meas{a,b}) \cdot \Ex(a,b)\\
\meas{\alpha \cdot a, \beta \cdot b} &= \alpha + \beta + \meas{a,b}
\end{align*}
Note that, on projects, the trefoil property becomes a Jacobi identity:
\[ \meas{\Ex(\alpha \cdot a, \beta \cdot b), \gamma \cdot c} = \meas{\alpha \cdot a, \Ex(\beta \cdot b, \gamma \cdot c)}. \] 
\end{definition}

Given any subset $\Perp \subset \Theta$, one can define an orthogonality relation from the measurement, by 
\[ \alpha \cdot a \perp \beta\cdot b \Leftrightarrow \meas{\alpha \cdot a, \beta \cdot b} \in \Perp. \]
One can then define \emph{types} as bi-orthogonally closed sets $A = A^{\bot\bot}$, where for any set $X$, 
\[ X^\bot = \{ \alpha\cdot a\mid \forall \xi\cdot x \in X, \xi \cdot x \perp \alpha\cdot a \}.\] 
Equivalently, a set $A$ is a type if and only if there exists another set $T$ (of \emph{tests}) such that $A=T^\bot$.

One can then show that the following constructions on types model the connectives of (multiplicative) linear logic:
\[
\begin{array}{rcl}
A\multimap B &=& \{ \xi\cdot x \mid \forall \alpha\cdot a\in A, \Ex(\xi\cdot x, \alpha\cdot a)\in B \}\\
A\otimes B &=& \{ \Ex(\alpha\cdot a, \beta\cdot b) \mid \alpha\cdot a \in A, \beta\cdot b\in B \}^{\bot\bot}
\end{array}
\]
One interesting fact is that the definition of $A\multimap B$ does not require the double orthogonal closure; it is nonetheless a type, as one can show that $A\multimap B = (A\otimes B^\bot)^\bot$. In general, models are moreover \emph{localised}, in the sense that each object has an assigned location, which can be understood as an associated point in a boolean algebra. Constructions such as ludics \citep{locussolum}, geometry of interaction \citep{multiplicatives,towards,goi1,goi2,goi3,goi5}, interaction graphs \citep{seiller-goim,seiller-goiadd,seiller-goig,seiller-goie,seiller-goif,seiller-markov}, or transcendental syntax \citep{syntran1,seiller-syntran} can be understood as examples of the above \citep{seiller-hdr}. In most of those cases, the models are extended to larger fragments of linear logic by considering additional operations on the underlying model of computation (the set of programs), which can be lifted to operations on types corresponding to logical connectives.

\section{Multiplicative linear logic}

From now on, we will use $\li$ instead of $\to$ for the implication in implicative structures. 

\subsection{Combinators}

Let us consider the following combinators : 
\[
\begin{array}{lll}
    \I = \lambda x.x & \hspace{2em} & \K = \lambda xy.x \\
    \B = \lambda xyz.x(yz) && \W = \lambda xy.xyy \\
    \C = \lambda xyz.xzy && \cS = \lambda xyz.xz(yz)
\end{array}
\]

\begin{proposition}[2.24]
We have the following equalities :
    \begin{align*}
        \I^\A & = \bigmeet_{a \in \A} (a \li a) \\
        \B^\A & = \bigmeet_{a,b,c \in \A} ((b \li c) \li (a \li b) \li a \li c) \\
        \C^\A & = \bigmeet_{a,b,c \in \A} ((a \li b \li c) \li b \li a \li c) \\
        \K^\A & = \bigmeet_{a,b \in \A} (a \li b \li a) \\
        \W^\A & = \bigmeet_{a,b \in \A} ((a \li a \li b) \li a \li b) \\
        \cS^\A & = \bigmeet_{a,b,c \in \A} ((a \li b \li c) \li (a \li b) \li a \li c)
    \end{align*}
\end{proposition}

The combinators $\K$, $\W$, and $\cS$ represent operations that delete or duplicate premises, and therefore should not be allowed without control in linear logic. This prevents us from including them in our definition of separator. So we restrict ourselves to a class of $\lambda$-terms large enough to represent the deduction rules of (multiplicative) linear logic without enclosing structural rules. 

\begin{definition}
    We will call ‘linear $\lambda$-term’ any $\lambda$-term where every abstraction binds exactly one variable and every free variable appears at most once.
\end{definition}

We will see that this is a large enough part of the $\lambda$-terms by showing that the typing rules we will use only produce such terms, and that it is small enough by giving an example of a linear model.

This subsection is dedicated to show that linear $\lambda$-terms are generated by $\I$, $\B$, and $\C$.

\begin{remark}
    Linear $\lambda$-terms are strongly normalisable.
\end{remark}


\begin{definition}
    A linear combinatory term is a linear $\lambda$-term which is either $\B$, $\C$, $\I$, a free variable or an application of linear combinatory terms.
\end{definition}

Our objective is to establish that each linear $\lambda$-term can be \mbox{$\beta$-expanded} to a linear combinatory term. In order to prove this theorem we will need some more combinators : 
$$
\begin{tabular}{ccc}
    $\I_1 = \I$ & $\quad$ &
    $\I_{n+1} = \B\I_n$ \\
    $\B_1 = \B$ & $\quad$ &
    $\B_{n+1} = \B\B_n$ \\
    $\C_1 = \C$ & $\quad$ &
    $\C_{n+1} = \B\C_n$ \\
\end{tabular}
$$
With some calculation one can establish that:
\begin{equation*}
    \I_n \bred \lambda \overline{x}.\overline{x} \quad
    \B_n \bred \lambda \overline{x}yz.\overline{x}(yz) \quad
    \C_n \bred \lambda \overline{x}yz.\overline{x}zy 
\end{equation*}
where $\overline{x} = x_1 \ldots x_n$.

We will now consider \(\lambda\)-terms representing permutations to establish the main theorem.

\begin{definition}
    Let $\sigma \in \perm_n$ be a permutation, we define the $\lambda$-term of this permutation as
    $$
    \lambda_\sigma := \lambda xy_1 \ldots y_n.xy_{\sigma^{-1}(1)} \ldots y_{\sigma^{-1}(n)}
    $$
\end{definition}

Let us consider two permutations $\sigma\in\perm_n$ and $\tau\in\perm_m$ that coincide on their shared support, that is either the restriction of $\sigma$ to $\{1,\dots,m\}$ is equal to $\tau$ or the restriction of $\tau$ to $\{1,\dots,n\}$ is equal to $\sigma$. Then $\sigma$ and $\tau$ can be represented by the same $\lambda$-term.

\begin{proposition}\label{prop:permutations}
   %
    Formally if $\sigma \in \perm_m$ and $\tau \in \perm_n$ are such that $m < n$, $\tau(i) = \sigma(i)$ for $i \leq n$, and  $\tau(i) = i$ for $i > n$, then for all $\lambda$-terms $t, t_1, ..., t_n$, 
    $$
        \lambda_\sigma t t_1 ... t_n =_\beta \lambda_\tau t t_1 ... t_n
    $$
\end{proposition}

This allows us to consider only $\sigma \in  \perm_n$ where $n$ is the greatest element in the support of $\sigma$, as $\lambda_\sigma$ will have the correct behavior even on bigger sets. In particular, we write permutations as products of cycles without worrying about the domains. 

\begin{proposition}\label{prop:combinatorypermutations}
    For each $\sigma \in \perm$, there is a closed linear combinatory term $t$ such that $t \bred \lambda_\sigma$.
\end{proposition}

\begin{proof}
    For all $\sigma, \tau \in \perm_n$, we can compute their composition: 
    \begin{align*}
        \B\lambda_\sigma \lambda_\tau 
            & \bred\lambda x. \lambda_\sigma (\lambda_\tau x) \\
            & \bred \lambda x y_1 \ldots y_n.\lambda_\tau x y_{\sigma^{-1}(1)} \ldots y_{\sigma^{-1}(n)} \\
            & \bred \lambda x y_1 \ldots y_n.x y_{\sigma^{-1}\tau^{-1}(1)} \ldots y_{\sigma^{-1}\tau^{-1}(n)} \\
            & \bred \lambda x y_1 \ldots y_n.x y_{(\tau\sigma)^{-1}(1)} \ldots y_{(\tau\sigma)^{-1}(n)} = \lambda_{\tau\sigma}.
    \end{align*}
    The result then comes from the fact that $\C_i = \lambda_{(i,i+1)}$, since the set $\{(i,i+1)\}_{i < n}$ generates $\perm_n$.
\end{proof}

We add the followings combinators to our list, for $l \leq k < n$ :
\begin{align*}
    \C_{l,k} & = \B\C_l(\B\C_{l+1}\ldots(\B\C_{k})\ldots) \\
    \B_{k,n} & = \underbrace{\B\B_k(\B\B_k\ldots(\B\B_k)\ldots)}_{n-k}\\
    \cmb_{k,n} & = \C_{2,k+1} \B_{k+1,n+1}
\end{align*}
and check that :
\begin{align*}
    \C_{l,k} & \bred \lambda_{(k+1, k,\ldots,l)} \\
    \B_{k,n} & \bred \lambda xy_1\ldots y_n.xy_1\ldots y_{k-1}(y_k \ldots y_n) \\
    \cmb_{k,n} 
    & \bred \lambda x y z_1 \ldots z_k.\B_{k+1,n+1} x z_1 \ldots z_k y\\
    & \bred \lambda x y z_1 \ldots z_n.x z_1 \ldots z_k (y z_{k+1} \ldots z_n)
\end{align*}

\begin{lemma}\label{lem:linearcombiterms}
    Let $t$ be a linear combinatory term. For all $x_1, \ldots, x_n \in \fv(t)$, there is a linear combinatory terms $t_0$ such that $t_0 \bred \lambda x_1 \ldots x_n.t$.
\end{lemma}

\begin{proof}
    The proof is by induction on the structure of $t$.
    If $t$ is a free variable, we take $\I$.
    If $t = uv$, then for all $x_1, \ldots, x_n \in \fv(t)$, $i \leq n$, either $x_i \in \fv(u)$ or $x_i \in \fv(v)$.
    By Proposition \ref{prop:combinatorypermutations}\footnote{Take a linear combinatory term $T \bred \lambda_\sigma$ and then use $(T\cmb_{k,n})$ instead of $\cmb_{k,n}$.} we can suppose that there is a $k \leq n$ such that $x_1, \ldots, x_k \in \fv(u)$ and $x_{k+1}, \ldots, x_n \in \fv(v)$.
    By induction, there are two linear combinatory terms $u_0 \bred \lambda x_1 \ldots x_k.u$ and $v_0 \bred \lambda x_{k+1} \ldots x_n.v$. Then we have
    \begin{align*}
        \cmb_{k,n}u_0v_0 
        & \bred \lambda x_1 \ldots x_n.u_0 x_1 \ldots x_k (v_0 x_{k+1} \ldots x_n) \\
        & \bred \lambda x_1 \ldots x_n.uv \\
        & = \lambda x_1 \ldots x_n.t
    \end{align*}
    If $t = \lambda y.u$ then for all $x_1, \ldots, x_n \in \fv(t)$, $x_1, \ldots, x_n, y \in \fv(u)$, so by induction there is a linear combinatory term $u_0 \bred \lambda x_1 \ldots x_n y.u = \lambda x_1 \ldots x_n.t$. 
\end{proof}

\begin{theorem}\label{thm:linearcombinatory}
    For each linear $\lambda$-term $t$, there is a linear combinatory term $t_0$ such that $t_0 \bred t$.
\end{theorem}

\begin{proof}
    The proof is by induction on the structure of $t$.
    If $t$ is a free variable, the result holds.
    If $t = uv$, then by induction there exists two linear combinatory terms $u_0,v_0$ such that $u_0 \bred u$ and $v_0 \bred v$, and $t_0 = u_0v_0$ is a linear combinatory terms such that $t_0 \bred t$.
    If $t = \lambda x.u$, by induction there exists a linear combinatory term $u_0 \bred u$, and the preceding lemma ensures that there exists an other linear combinatory term $t_0 \bred \lambda x.u_0 \bred t$.
\end{proof}

\subsection{Separation}

We now suppose given an implicative algebra. We will define the notion of linear separator.

\begin{definition}
    A \emph{linear separator} is a subset $S \subseteq \A$ such that:
    \begin{enumerate}
        \item If $a \in S$ and $a \preq b$ then $b \in S$; \label{it1}
        \item $\I^\A, \B^\A, \C^\A \in S$;
        \item If $a \in S$ and $(a \li b) \in S$ then $b \in S$. \label{it3}
    \end{enumerate}
    We say that $S$ is consistent if $\bot \not\in S$.
\end{definition}

\begin{remark}
    Item \ref{it1} allow us to reformulate item \ref{it3} as:
    \begin{enumerate}
        \item[3'.] If $a \in S$ and $b \in S$ then $ab \in S$.
    \end{enumerate}
\end{remark}


    Let us note that any separator in the sense of Miquel \citep{DBLP:journals/mscs/Miquel20} is in particular a linear separator. It should be clear, however, that the converse does not hold.

\begin{proposition}[linear $\lambda$-closure]\label{prop:lambdaclosure}
    If $S \subseteq \A$ is a linear separator, then for all linear $\lambda$-term $t$ with free variables $\overline{x}$ and for all parameters $\overline{a} \in S$:
    \[
        (t\{\overline{x} := \overline{a}\})^\A \in S
    \]
    In particular for all closed linear $\lambda$-term $t$, $t^\A \in S$.
\end{proposition}

(TODO: keep the following as a separate definition, or push it into the def. of linear separator?)

\begin{definition}
    We call linear separator generated by $X$ and write $\lsep(X)$ the smallest linear separator containing $X$, and we define the linear core of $\A$ as $\lcore(\A) := \lsep(\emptyset)$.
\end{definition}

\begin{remark}
    It is straightforward from \autoref{thm:linearcombinatory} that 
    \[ \lsep(X) = \uacl(X \cup \{\I^\A, \B^\A, \C^\A\}), \] were $@$ denote the applicative closure, and $\uparrow$ the closure w.r.t. $\beta$-expansion (which we will sometimes call \emph{upward closure}).
\end{remark}

\begin{lemma}
    For every linear separator $S \subseteq \A$:
    \[
        (a \li b) \in S \implies b \in \lsep(S \cup \{a\}).
    \]
\end{lemma}


Note that we only have an implication here, because $b \in \lsep(S \cup \{a\})$ means that $a$ can be deduced from formulas in $S \cup \{a\}$. It does not make any assumption on how many time each formula is used.
This illustrate a fundamental difference between contexts and theories in linear logic, that does arise in the context of intuitionnistic or classical logic.

\subsection{Interpreting MLL}

We now focus on the interpretation of multiplicative linear logic. We restrict to the case of intuitionnistic linear logic, as the classical variant requires additional structure described in the remark below.

\begin{definition}
Let $(\A, \preq, \li)$ be a linear implicative structure. We define:
\[
    a \tens b := \bigmeet_{c \in \A} ((a \li b \li c) \li c).
\]
\end{definition}

One can then show that this tensor product validates the multiplicative rules for conjunction.

\begin{proposition}\label{prop:semantictyping:tensor}
    The following semantic typing inferences are valid in any linear implicative structure:
    \[
        \AxiomC{$\Gamma \vdash t : a$}
        \AxiomC{$\Delta \vdash u : b$}
        \RightLabel{$\tens$-R}
        \BinaryInfC{$\Gamma,\Delta \vdash \lambda z.ztu : a \tens b$}
    \DisplayProof
    \quad
        \AxiomC{$\Gamma, x : a, y : b \vdash t : c$}
        \RightLabel{$\tens$-L}
        \UnaryInfC{$\Gamma, z : a \tens b \vdash z(\lambda xy.t) : c$}
    \DisplayProof
    \]
\end{proposition}

We can thus interpret intuitionnistic multiplicative linear logic (IMLL) without unit inside an implicative structure $(\A, \preq, \li)$. For this, we will use the following syntax.

Formulas are defined from a set of variables $\varset$ by the following grammar:
\[
    A ::= X \mid A \li A \mid A \tens A, \quad (X\in\varset).
\]
The deduction rules for IMLL are given in Figure \autoref{fig:IMLL}. 

\begin{figure}
\[
\begin{array}{cc}
    \AxiomC{}
    \RightLabel{ax}
    \UnaryInfC{$A \vdash A$}
\DisplayProof
&
    \AxiomC{$\Gamma \vdash A$}
    \AxiomC{$\Delta, A \vdash C$}
    \RightLabel{cut}
    \BinaryInfC{$\Gamma,\Delta \vdash C$}
\DisplayProof
\\ \ \\
    \AxiomC{$\Gamma \vdash A$}
    \AxiomC{$\Delta \vdash B$}
    \RightLabel{$\tens$-R}
    \BinaryInfC{$\Gamma,\Delta \vdash A \tens B$}
\DisplayProof
&
    \AxiomC{$\Gamma, A, B \vdash C$}
    \RightLabel{$\tens$-L}
    \UnaryInfC{$\Gamma, A \tens B \vdash C$}
\DisplayProof
\\ \ \\
    \AxiomC{$\Gamma, A \vdash B$}
    \RightLabel{$\li$-R}
    \UnaryInfC{$\Gamma \vdash A \li B$}
\DisplayProof
&
    \AxiomC{$\Gamma \vdash A$}
    \AxiomC{$\Delta, B \vdash C$}
    \RightLabel{$\li$-L}
    \BinaryInfC{$\Gamma,\Delta, A \li B \vdash C$}
\DisplayProof
\end{array}
\]
\caption{Rules for IMLL}\label{fig:IMLL}
\end{figure}

We now define the interpretation of IMLL formulas, and then show a soundness result.

\begin{definition}
An interpretation is a function which associates a value $X^\A \in \A$ to every atomic formula $X$, and then is recursively defined on every formula by $(A \li B)^\A = A^\A \li B^\A$ and $(A \tens B)^\A = A^\A \tens B^\A$.
\end{definition}

\begin{proposition}[Soundness]\label{prop:MLLsoundness}
    If a formula $A$ is a tautology in IMLL, then $A^\A \in \lcore(\A)$.
\end{proposition}

\begin{remark}
While we have focussed on IMLL here, a natural extension of the above definitions provides a framework for a sound interpretation of MLL. More formally, one can consider the additional definitions:
\begin{align*}
    \orth a &:= a \li \bot \\
    a \parr b &:= \bigmeet_{c \in \A} ((a \li c) \li (b \li c) \li c) \\
    \1 &:= \bot \li \bot
\end{align*}
Then, by adding an element of type $\bigmeet_{a \in \A} a \li a^{\bot\bot}$ inside the separator (as it is done with $cc$ in classical separators), we could fully interpret MLL.
\end{remark}

\subsection{Quotient by a linear separator}

One important construction in the setting of implicative algebras is the consideration of quotients. These define Heyting algebras in the case of Miquel, providing a generic construction of a topos from any implicative algebras. More precisely, implicative algebras define tripoi, which by the tripos-to-topos construction give rise to topoi.

Here, one cannot expect to obtain a Heyting algebra when quotienting by the entailment relation. Although one may expect the structure of a quantale to arise, this is not the case in general as the quotiented structure does not possess all the required limits. The resulting structure resemble that of a residuated lattice.

\begin{definition}
    Let $(\A, \preq, \li)$ be an implicative structure, and $S$ a linear separator. The entailment relation $\ent_S$ is defined as:
    \[
        a \ent_S b \iff a \li b \in S.
    \]
    This is a preorder, and we denote by $\enteq_S$ the induced equivalence relation.
\end{definition}

We will now consider $\A/S := \A/\enteq_S$, the quotient of $\A$ w.r.t. the entailment equivalence.

\begin{proposition}\label{prop:quotientdef}
    The following operations are well defined on $\A/S$ : 
    \begin{align*}
        \cl a \li \cl b & := \cl{a \li b}  \\
        \cl a \tens \cl b & := \cl{a \tens b}
    \end{align*}
\end{proposition}

\begin{proposition}\label{prop:quotientproperties}
    If $S$ is a linear separator, theses operations have the following properties :
    \begin{enumerate}
        \item commutativity: $\cl a \tens \cl b = \cl b \tens \cl a$;
        \item associativity: $(\cl a \tens \cl b) \tens \cl c = \cl a \tens (\cl b \tens \cl c)$;
        \item identity is neutral: $\cl a \tens \cl{\I^\A} = \cl a$;
        \item currying: $\cl a \tens \cl b \li \cl c = \cl a \li \cl b \li \cl c$;
        \item modus ponens: $(\cl a \li \cl b) \tens \cl a \ent_S \cl b$; 
        \item pairing: $\cl a \ent_S \cl b \li \cl a \tens \cl b$;
        \item transitivity: $(\cl a \li \cl b) \tens (\cl b \li \cl c) \ent_S \cl a \li \cl c$;
        \item tensor rule: $(\cl a \li \cl b) \tens (\cl c \li \cl d) \ent_S \cl a \tens \cl c \li \cl b \tens \cl d$.
    \end{enumerate}
\end{proposition}

\subsection{Linear implicative algebras from linear realisability situations}

We finally establish that models defined from linear realisability situations give rise to linear implicative algebras.

\begin{definition}
    We call linear implicative algebra any linear implicative structure equipped with a linear separator $S \subseteq \A$.
    A linear implicative algebra $(\A, \preq, \li, S)$ is said consistent if $S$ is consistent.
\end{definition}

We first prove a general result, showing that under a few assumptions, linear realisability situations give rise to linear implicative algebras. We will then explain why these additional assumptions are reasonable: all linear realisability models from the literature satisfy those\footnote{All but the most recent construction \citep{LogicNucleus}, in which commutativity is not assumed (with a motivating example which is not commutative). However, this model does not soundly model linear logic but a substructural logic with a non-commutative multiplicative conjunction.}.

We will proceed step by step, which allows us to pinpoint where exactly the hypotheses are used, namely in showing that the set of non-empty types is a coherent separator. A first easy result is that the set of types defined from a linear realisability situation is a lattice.

\begin{proposition}\label{prop:linreallattice}
    Given a linear realisability situation $(P,\ex,\meas{\cdot,\cdot})$, the structure $(\types, \subseteq)$ is a complete lattice with:
    \[
    \begin{array}{l}
        \bigmeet_{i \in I} \mathbf{A}_i = \bigcap_{i \in I} \mathbf{A}_i\\
        \bigjoin_{i \in I} \mathbf{A}_i = (\bigcup_{i \in I} \mathbf{A}_i)^{\simperp\simperp}
    \end{array}
    \]
\end{proposition}

(TODO: Why the \(\simperp\)?)

This lattice gives us the underlying structure to define the linear implicative algebra. We now can check that the linear implication connective, defined on types, makes the lattice of types an applicative structure. This, again, does not require additional hypotheses.

\begin{proposition}\label{prop:linrealappstruct}
    Given a linear realisability situation $(P,\ex,\meas{\cdot,\cdot})$, the tuple $(\types, \subseteq, \li)$ is an applicative structure with $\Ex$ as application.
\end{proposition}

\begin{proof}[Proof of Proposition \ref{prop:linrealappstruct}]
    We simply check that all needed axioms are satisfied. For axiom 1, let $\mathbf{A}, \mathbf{A}', \mathbf{B} , \mathbf{B}' \in \types$ such that $\mathbf{A}' \subseteq \mathbf{A}$ and $\mathbf{B} \subseteq \mathbf{B}'$. We have :
    \begin{align*}
    \proj{p} \in \mathbf{A} \li \mathbf{B}
    & \implies \forall \proj{a} \in \mathbf{A}, \ex(\proj{p}, \proj{a}) \in \mathbf{B}  \\
    & \implies \forall \proj{a} \in \mathbf{A}', \ex(\proj{p}, \proj{a}) \in \mathbf{B}'  \\
    & \implies \proj{p} \in \mathbf{A}' \li \mathbf{B}'
    \end{align*}

    For axiom 2, let $\mathbf{B}_i \in \types$ for all $i \in I$, $\mathbf{A} \in \types$. We have :
    \begin{align*}
        \bigcap_{i \in I} (\mathbf{A} \li \mathbf{B}_i)
        & = \bigcap_{i \in I} \{\proj{p} \in \Proj : \forall \proj{a} \in \mathbf{A}, \ex(\proj{p}, \proj{a}) \in \mathbf{B}_i\}  \\
        & = \{\proj{p} \in \Proj : \forall \proj{a} \in \mathbf{A}, \ex(\proj{p}, \proj{a}) \in \bigcap_{i \in I} \mathbf{B}_i\}  \\
        & = \mathbf{A} \li (\bigcap_{i \in I} \mathbf{B}_i)
    \end{align*}

    For application, let $\mathbf{A}, \mathbf{B} \in \types$ :
    \begin{align*}
        \mathbf{A}\mathbf{B}
        & = \bigcap \{\mathbf{C} \in \types : \mathbf{A} \subseteq \mathbf{B} \li \mathbf{C}\} \\
        & = \bigcap \{\mathbf{C} \in \types : \mathbf{A} \subseteq \{\proj{p} \in \Proj : \forall \proj{b} \in \mathbf{B}, \ex(\proj{p}, \proj{b}) \in \mathbf{C}\}\} \\
        & = \bigcap \{\mathbf{C} \in \types : \forall \proj{a} \in \mathbf{A}, \forall \proj{b} \in \mathbf{B}, \ex(\proj{a}, \proj{b}) \in \mathbf{C}\} \\
        & = \{\ex(\proj{a}, \proj{b}) : \forall \proj{a} \in \mathbf{A}, \forall \proj{b} \in \mathbf{B}\} \\
        & = \mathbf{A} :: \mathbf{B}\qedhere
    \end{align*}
\end{proof}

We now need to show the existence of a coherent linear separator. This is where additional hypotheses are required, which is not surprising: while the abstract notion of linear realisability situation ensures that one can define types and operations between them, it does not require the existence of specific terms, such as identities. The additional hypotheses added here simply ensure that those elementary terms exist. It should thus not be a surprise that all known instances of linear realisability models do satisfy those additional requirements.

\begin{proposition}\label{prop:linrealsep}
    Suppose given a linear realisability situation $(P,\ex,\meas{\cdot,\cdot})$. If the measurement is symmetric, and there exists $id, \tau \in P$ such that for all $p, q \in P$ : 
    \begin{align*}
        \ex(id, p) & = p   &   \m{id}{p} & = 0 \\
        \ex(\tau, \ex(p, q)) & = \ex(q, p)   &   \m{\tau}{p} & = 0
    \end{align*}
    then $\types \setminus \emptyset$ is a coherent separator.
\end{proposition}

\begin{proof}
    We first check coherence, which is straightforward:
    $\bot = \emptyset \not\in \types \setminus \emptyset$, so if $\types \setminus \emptyset$ is a linear separator, it is coherent. 
    Now, let us show that it is indeed a linear separator.
    
    The fact that it is upward close is clear. 
    For applicative closure, by definition if $\mathbf{A}, \mathbf{B} \in \types$ are non empty, 
    then $\mathbf{A} :: \mathbf{B}$ is non empty.
     As $\ex$ is associative, for clarity and readability, we will use the infixed notation $\cdot$
     instead of $\ex$ until the end of this proof.

    \textbf{We now check that it contains $\I^\types$}. We write $\pid = (0, id)$, where $0$ is the neutral element of $\Theta$. 
    For all $\proj{p} = (a, p) \in \Proj$, we have:
    \[
        \pid \cdot \proj{p} = (0 + a + \m{id}{p}, id \cdot p) = (a, p) = \proj{p}
    \]
    and therefore:
    \begin{align*}
        \forall \proj{p} \in \Proj, \pid \cdot \proj{p} = \proj{p}  
        & \implies \forall \mathbf{A} \in \types, \pid \in \{\proj{p} \in \Proj : \forall \proj{a} \in \mathbf{A}, \proj{p} \cdot \proj{a} \in \mathbf{A}\}  \\
        & \implies \forall \mathbf{A} \in \types, \pid \in \mathbf{A} \li \mathbf{A}  \\
        & \implies \pid \in \bigcap_{\mathbf{A} \in \types} (\mathbf{A} \li \mathbf{A}) = \I^\types.
    \end{align*}
    
    \textbf{We now check that it contains $\B^\types$.}
    Let $\mathbf{A}, \mathbf{B}, \mathbf{C} \in \types$,
    we suppose $\mathbf{B} \li \mathbf{C}$, $\mathbf{A} \li \mathbf{B}$, and that $\mathbf{A}$ is non empty. Then for all $\proj{f} \in \mathbf{B} \li \mathbf{C}$, $\proj{g} \in \mathbf{A} \li \mathbf{B}$, $\proj{a} \in \mathbf{A}$, we have:
    \[
        ((\pid \cdot \proj{f}) \cdot \proj{g}) \cdot \proj{a} = \proj{f} \cdot (\proj{g} \cdot \proj{a}) \in \mathbf{C}.
    \]
    If $\mathbf{A}$ is empty, $\mathbf{A} \li \mathbf{C} = \top\ (= \types)$, hence 
    \[ (\mathbf{B} \li \mathbf{C}) \li (\mathbf{A} \li \mathbf{B}) \li \mathbf{A} \li \mathbf{C} = \top. \]
    If $\mathbf{A} \li \mathbf{B}$ is empty, $(\mathbf{A} \li \mathbf{B}) \li \mathbf{A} \li \mathbf{C} = \top$ so 
    \[ (\mathbf{B} \li \mathbf{C}) \li (\mathbf{A} \li \mathbf{B}) \li \mathbf{A} \li \mathbf{C} = \top. \]
     If $\mathbf{B} \li \mathbf{C}$ is empty, 
     \[ (\mathbf{B} \li \mathbf{C}) \li (\mathbf{A} \li \mathbf{B}) \li \mathbf{A} \li \mathbf{C} = \top. \]

    In all fours cases above, we can conclude that 
    \[
        \pid \in (\mathbf{B} \li \mathbf{C}) \li (\mathbf{A} \li \mathbf{B}) \li \mathbf{A} \li \mathbf{C},
    \]
    which proves that 
    \[
        \pid \in \bigcap_{\mathbf{A}, \mathbf{B}, \mathbf{C} \in \types} ((\mathbf{B} \li \mathbf{C}) \li (\mathbf{A} \li \mathbf{B}) \li \mathbf{A} \li \mathbf{C}) = \B^\types.
    \]
    
    \textbf{Finally, we check that it contains $\C^\types$.} We write $\proj{t} = (0, \tau)$. 
    For all $\proj{p} = (a, p), \proj{q} = (b, q) \in \Proj$, we have :
    \begin{align*}
        \proj{t} \cdot (\proj{p} \cdot \proj{q}) & = (0 + \m{\proj{p}}{\proj{q}} + \m{\tau}{p \cdot q}, \tau \cdot (p \cdot q)) \\
        & = (\m{\proj{p}}{\proj{q}}, q \cdot p) = \proj{q} \cdot \proj{p}.
    \end{align*}
    Let $\mathbf{A}, \mathbf{B}, \mathbf{C} \in \types$.
    We suppose $\mathbf{A}$, $\mathbf{B}$ and $\mathbf{A} \li \mathbf{B} \li \mathbf{C}$ non empty. Then for all $\proj{f} \in \mathbf{A} \li \mathbf{B} \li \mathbf{C}$, $\proj{a} \in \mathbf{A}$, $\proj{b} \in \mathbf{b}$, we have :
    \begin{align*}
        \proj{t} \cdot \proj{t} \cdot \proj{f} \cdot \proj{b} \cdot \proj{a} 
        & = \proj{f} \cdot \proj{t} \cdot \proj{b} \cdot \proj{a}  \\
        & = \proj{f} \cdot \proj{a} \cdot \proj{b} \in \mathbf{C}
    \end{align*}
    so 
    \[
        \proj{t} \cdot \proj{t} \in (\mathbf{A} \li \mathbf{B} \li \mathbf{C}) \li \mathbf{B} \li \mathbf{A} \li \mathbf{C}.
    \]
    If $\mathbf{A}$, $\mathbf{B}$ or $\mathbf{A} \li \mathbf{B} \li \mathbf{C}$ are empty we follow a similar argument as above, and we can conclude:
    \[
        \proj{t} \cdot \proj{t} \in \bigcap_{\mathbf{A}, \mathbf{B}, \mathbf{C} \in \types} ((\mathbf{A} \li \mathbf{B} \li \mathbf{C}) \li \mathbf{B} \li \mathbf{A} \li \mathbf{C}) = \C^\types.\qedhere
    \]
\end{proof}

We now have all the elements to establish the main theorem of this section.
\begin{theorem}\label{thm:lineareal}
    Let $(P,\ex,\meas{\cdot,\cdot})$ be a linear realisability situation such that the measurement $\ex$ is symmetric, and in which there exists elements $id, \tau \in P$ such that for all $p, q \in P$: 
    \[
    \begin{array}{lll}
        \ex(id, p) = p,  &\hspace{2em}&   \m{id}{p} = 0, \\
        \ex(\ex(\tau, p), q) = \ex(q, p),   &&   \m{\tau}{p} = 0.
    \end{array}
    \]
    Then $(\types, \subseteq, \li, \types \setminus \emptyset)$, where $\types$ is the set of all types, is a consistent linear implicative algebra.
\end{theorem}


As already mentioned, the three additional hypotheses seem restrictive but in all models from the literature mentioned above the execution and measurement are symmetric. And in this case, the conditions boil down to the existence of an identity $id$ such that $\m{id}{p}=0$ for all $p$. Note however that this comes from the fact that models from the literature are \emph{localised} (i.e. objects are assigned an element from a boolean algebra), but non-localised models can be constructed \citep{seiller-hdr}, in which the measurement is symmetric but the execution need not be. 

\begin{remark}
As recently described \citep{LogicNucleus}, the constructions of linear realisability can be performed in the case of a non-symmetric measurement and non-commutative execution $\Ex$. However, this requires adaptations and model a different substructural logic.
\end{remark}



As explained by Seiller \citep{seiller-hdr}, linear logic models defined by Geometry of Interaction \citep{multiplicatives,goi1,goi2,goi3,goi5}, ludics \citep{locussolum}, Interaction Graphs \citep{seiller-goim,seiller-goiadd,seiller-goig,seiller-goie,seiller-goif,seiller-markov}, and transcendental syntax \citep{syntran1,syntran2,syntran3,seiller-syntran} can be understood as models induced by a linear realisability situation. As a consequence, the set of types in any of these models (which sometimes are named differently, such as \emph{conducts}, \emph{behaviors}, etc.) define a linear implicative algebra.

Most of those models however encompass larger fragments of linear logic. We will now consider extensions of linear implicative algebras accounting for larger fragments.

\section{Exponentials}
\label{sec:exponentials}

In this section, we explain a natural extension of the above framework to interpret exponential connectives. This approach follows the construction of exponential connectives in linear realisability models. 

While this will not be detailed here, it is not difficult to see that additive connectives can be defined by the lattice structure. Note that this definition of additives mirrors the definition of additive connectives in some of the linear realisability models mentioned above, in particular Girard's geometry of interaction in the hyperfinite factor \citep{goi5} and Interaction Graphs \citep{seiller-goiadd,seiller-goig,seiller-goie,seiller-goif,seiller-markov}.

The contents of the current section can thus be extended to the additive fragment. We will however detail in the next section another, more structured, approach to additive connectives, which connects to other realisability models (such as older geometry of interaction models \citep{goi3}) and can be used to define exponential connectives using the formula for the free exponential \citep{freeexp}.


\subsection{Implicatives structures with exponentials}

\begin{definition}
    We call exponential implicative structure an implicative structure $(\A, \preq, \li)$ with a monotonic unary operation $\oc : \A \to \A$.
\end{definition}

We enrich $\lambda$-calculus by adding $\oc$ and a new $\oc$-abstraction, so we have :
\[ 
    t,u ::= x \mid tu \mid \lambda x.t \mid \oc t \mid \lambda \oc x.t
\]
We call them $\lambda_\oc$-terms and we equip them with the $\oc$-reduction : 
\[
    (\lambda \oc x.t)!u \bangred t[x : u]
\]
We expand the interpretation of $\lambda$-terms to those new terms by adding the following induction steps.
\[
\begin{array}{rcl}
    (!t)^\A & := & !(t^\A) \\
    (\lambda !x.t)^\A & := & \bigmeet_{a \in \A} (!a \li (t[x:a])^\A)
\end{array}
\]

\begin{proposition}[Monotony of substitution] \label{monotony of substitution with !}
    For every $\lambda_\oc$-term $t$ with free variables $x_1, \ldots, x_n$, and for all parameters
    $a_1 \preq a_1'$, $\ldots$, $a_n \preq a_n'$, we have:
    \[
        (t\{x_1 := a_1, \ldots, x_n := a_n\})^\A \preq (t\{x_1 := a_1', \ldots, x_n := a_n'\})^\A.
    \]
\end{proposition}

\begin{proposition}\label{reductions2}
    If $t \bangred u$ then $t^\A \preq u^\A$.
\end{proposition}

For any type $a \in \A$, we allow ourselves to add in typing contexts statements of the form $\oc x : \oc a$.
As a substitution it is a synonym for $x : a$, but it will allow us to keep track of exponentials inside $\lambda_!$-terms.

\begin{proposition}
    The following typing rules are valid in any exponential implicative structure:
    \[
        \AxiomC{$\Gamma, !x : !A \vdash t : C$}
        \RightLabel{$\li R'$}
        \UnaryInfC{$\Gamma \vdash \lambda !x.t : !A \li C$}
    \DisplayProof
    \]
    \[
    \begin{array}{ccc}
        \AxiomC{$\Gamma \vdash t : C$}
        \RightLabel{!wL}
        \UnaryInfC{$\Gamma, !x : !A \vdash (\lambda !x.t)!x : C$}
    \DisplayProof
    &
    &
        \AxiomC{$\Gamma, !x : !A, !y : !A \vdash t : C$}
        \RightLabel{!cL}
        \UnaryInfC{$\Gamma, !z : !A \vdash t[x:z, y:z] : C$}
    \DisplayProof
    \\~\\
        \AxiomC{$\Gamma, x : A \vdash t : C$}
        \RightLabel{!dL}
        \UnaryInfC{$\Gamma, !x : !A \vdash t : C$}
    \DisplayProof
    &&
        \AxiomC{$!\Gamma \vdash t : A$}
        \RightLabel{!R}
        \UnaryInfC{$!\Gamma \vdash !t : !A$}
    \DisplayProof
    \end{array}
    \]
\end{proposition}

\begin{proof}
    It suffices to compute the interpretation of the $\lambda_!$-terms.
\end{proof}

We now have two different rules $\li_\oc R$ and $\li_\oc R'$ that overlap. To prevent this, we will always use statements of the form $\oc x : \oc a$ instead of $x : \oc a$. We may think that $\oc x : \oc a$ is a stronger statement as it suppose the existence of an element of type $a$, but as we use generalized realizers, it is always true.

We consider the following additional combinators. The \emph{elementary group}:
    \begin{align*}
    \Ke & := \lambda x \oc y.x  \\
    \We & := \lambda x \oc y.x \oc y \oc y  \\
    \F  & := \lambda \oc x \oc y. \oc (xy)
    \end{align*}
and the \emph{comonadic group}:
\begin{align*}
     \D  & := \lambda \oc x.x  \\
    \cd & := \lambda \oc x. \oc \oc x  
\end{align*}

\begin{proposition}
    The following equalities are satisfied:
    \begin{align*}
        \Ke^\A & = \bigmeet_{a,b \in \A} a \li !b \li a  \\
        \We^\A & = \bigmeet_{a,b \in \A} (!a \li !a \li b) \li !a \li b  \\
        \F^\A  & = \bigmeet_{a,b \in \A} !(a \li b) \li !a \li !b \\
        \D^\A  & = \bigmeet_{a \in \A} !a \li a  \\
        \cd^\A & = \bigmeet_{a \in \A} !a \li !!a  
    \end{align*}
\end{proposition}

\begin{definition}
    We will call \emph{linear $\lambda_\oc$-term} any $\lambda_\oc$-term where every $\lambda$-abstraction bind exactly one variable, this variable is not under a $\oc$, and every free variable appears at most once.
\end{definition}

\begin{definition}
    An exponential combinatory term is a linear $\lambda$-term which is either $\B$, $\C$, $\I$, $\Ke$, $\We$, $\D$, $\cd$, $\F$, a free variable, the $\oc$ of one of the previous, or an application of exponential combinatory terms.
\end{definition}

\begin{theorem}\label{exponential combinatory}
    For each linear $\lambda_\oc$-term $t$, there is a exponential combinatory term $t_0$ such that $t_0 \red t$.
\end{theorem}

\begin{proposition}
    The following typing rules are now valid, where $\Gamma = x_1 : a_1, ..., x_n : a_n$ in $!R$.
    \begin{equation*}
    \begin{array}{c}
        \AxiomC{$\Gamma \vdash t : C$}
        \RightLabel{!wL}
        \UnaryInfC{$\Gamma, !x : !A \vdash \Ke t !x : C$}
    \DisplayProof
    \\\\
    \AxiomC{$\Gamma, !x : !A, !y : !A \vdash t : C$}
        \RightLabel{!cL}
        \UnaryInfC{$\Gamma, !z : !A \vdash \We (\lambda !x!y.t) !z : C$}
    \DisplayProof
    \\\\
        \AxiomC{$\Gamma, x : A \vdash t : C$}
        \RightLabel{!dL}
        \UnaryInfC{$\Gamma, !x : !A \vdash (\lambda x.t) (\D !x) : C$}
    \DisplayProof
    \\\\
        \AxiomC{$!\Gamma \vdash t : A$}
        \RightLabel{!R}
        \UnaryInfC{$!\Gamma \vdash \F (... (\F !(\lambda !x_1, ..., !x_n.t) !x_1) ...)!x_n : !A$}
    \DisplayProof
    \end{array}
    \end{equation*}
\end{proposition}

\begin{definition}
    We say that a linear separator is an exponential separator if it contains the combinators $\Ke, \We, \D, \cd, \F$, and is closed by $!$. We call exponential core (written $\ecore$) the smallest exponential separator.
    
    A linear separator closed by $\oc$ and containing the combinators $\Ke, \We, \F$ is called an \emph{elementary} separator. The elementary core, written $\elemcore$, is the smallest elementary seprarator.
\end{definition}

\begin{proposition}
    If a formula $A$ is a tautology in IMELL, then $A^\A \in \ecore(\A)$. If a formula $A$ is a tautology in IMELL provable without dereliction and digging, then $A^\A \in \elemcore(\A)$.
\end{proposition}

\begin{theorem}\label{thm:fromlineartoint}
    Let $(\A, \preq, \li, !)$ be an exponential implicative structure. Then $(\A', \preq, \to)$, with $\A' = \A$ and $a \to b =\ !a \li b$, is an implicative structure.

    In addition, if $S$ is an exponential separator in $\A$, it is an intuitionistic separator in $\A'$.
\end{theorem}

\begin{proof}
    The two axioms are easy to check.
    Upward closure is the same in both structure.
    Closure by modus ponens follow from closure by exponential.
    We have
    \begin{align*}
        (\B\Ke\D)^\A & \preq \K^{\A'} \\
        (\lambda !x!y!z.x!z!(y!z))^\A & \preq \cS^{\A'} 
    \end{align*}
    with 
\[\cS^{\A'} = \bigmeet_{a,b,c \in \A} (!(!a \li !b \li c) \li !(!a \li b) \li !a \li c)\]\qedhere
\end{proof}

\subsection{Exponentials in linear realisability models}

The above approach to exponential can be used to show that some of the 
known linear realisability constructions are not only instances of linear 
implicative algebras but also of either elementary or exponential implicative 
algebras.

(TODO: in the following, we should add intersections since the proof-term is universal, but it should be introduced with a footnote since intersection is in the papers localised.)

In particular, Girard's geometry of interaction in the hyperfinite factor \citep{goi5}
and Seiller's interaction graphs model based on graphings \citep{seiller-goig} are 
instances of elementary implicative algebras. Indeed, both models introduce a
operation called \emph{perennisation}, noted $\oc$, on the set of types. In both cases,
it is established that specific terms inhabit the following types:
\[ \mathbf{\bigcap_{A,B}\oc(A\multimap B)\otimes \oc A \multimap \oc B}, \quad \mathbf{\bigcap_{A,B} A\otimes \oc B \multimap A}, \quad \mathbf{\bigcap_{A} \oc A \otimes \oc A \multimap \oc A},   \]
which, up to de Morgan equalities and standard manipulation, establishes that $\Ke$, $\We$, and $\F$ are inhabited.
It is worth noting that some arguments are given \citep{seiller-goig} to establish that those models do not satisfy further 
exponential principles, and therefore should not define exponential implicative algebras.

Other linear realisability models, notably the first geometry of interaction models \citep{goi1,goi2,goi3} and Seiller's later extensions of Interaction Graphs models \citep{seiller-goif,seiller-markov} give rise to exponential implicative algebras. Indeed, it is established that the above three types are inhabited, but it is also the case of the following two:
\[ \mathbf{\bigcap_{A} \oc A \multimap \oc \oc A}, \quad \mathbf{\bigcap_{A} \oc A \multimap  A}.\]
These results thus establish that in these models, $\Ke$, $\We$, $\F$, $\D$, and $\cd$ are inhabited. As a consequence, those models define implicative algebras in the sense of Miquel based on Theorem \ref{thm:fromlineartoint}, something that was not known before.



\section{Additives and the fixpoint exponential}

In this section, we will introduce additional structure to define additive connectives in a different, more explicit way. This will allow us to introduce an alternative approach for exponential connectives based on the fixpoint exponential construction.

\subsection{Records}

\begin{definition}[record]
  Let $(\A, \preq, \to)$ be an implicative structure. We say that a
  \emph{record} on $\A$ is a function $\Recl: \A \to \A$ such
  that:
  \begin{align*}
    \bigmeet_{i \in I} \Recl(a_i) = \Recl\left(\bigmeet_{i \in I}a_i\right)
  \end{align*}
\end{definition}

\begin{remark}
The above definition implies the existence of a second function such that:
    \begin{align*}
      \Appl{a} \preq b\qquad\qquad &\Leftrightarrow \qquad\qquad a \preq
                                       \Recl(b).
    \end{align*}
Both $\Appl$ and $\Recl$ are monotone and satisfy the following:
\begin{align*}
\Recl(\top)&= \top\\
\Appl{a} &= \bigmeet \{ b \in \A \mid a \preq \Recl(b) \}.
\end{align*}
 This adjoint couple induces both a monad and a comonad on $\A$, hence $a
    \preq \Recl(\Appl{a})$ and $\Appl{\Recl(a)} \preq a$.
\end{remark}

\begin{proposition}[typing rules for records]
  Let $\Recl$ be a record. The semantic typing rules
  \begin{align*}
    \infer{\Gamma \vdash \Appl{t} : a}{\Gamma\vdash t : \alpha} \qquad \qquad
    \infer{\Gamma \vdash \Recl(t) : \Recl(a)}{\Gamma\vdash t : a}
  \end{align*}
  are valid in any implicative structure, where $\alpha$ is such that
  $\Appl{\alpha} = \Appl{\Recl(a)}$.
\end{proposition}

We say that two records $\Recl$ and $\Recr$ are compatible if
\begin{align*}
  \Appl{(\Recl(a) \curlywedge \Recr(b))} &= \Appl{\Recl(a)}\\
  \Appr{(\Recl(a) \curlywedge \Recr(b))} &= \Appr{\Recr(b)}
\end{align*}

\begin{example}[$\lambda$-calculus with records]
  \label{ex:lambda-calculus-records}
  Given disjoints sets of \emph{variables} $x,y,\dots$ and \emph{labels}
  $\ell_1,\ell_2,\dots$, the $\lambda$-calculus with records is defined by the
  grammar:
  \begin{align*}
    t,u := x \mid tu \mid \lambda x. t \mid t \mid \Record{} \mid
    \Record{t ; \ell = u} \mid t . \ell
  \end{align*}
  with the reduction rules
  \begin{align*}
    (\lambda x. t)u &\to_{\beta} t[u/x]\\
    \Record{t ; \ell =u}.\ell &\to_{\iota} u\\
    \Record{t; \ell=u}.\ell' &\to_{\iota} t.\ell'\text{ if }\ell\neq \ell'
  \end{align*}

  We will write $\Record{\ell_1 = t_1 ; \dots ; \ell_n = t_n}$ for the iterated
  nesting 
  \[ \Record{\Record{\cdots\Record{} ; \ell_1 = t_1} \cdots ; \ell_n =
    t_n}.\] 
    Note that the operational semantics given by the $\iota$-reduction
  means that the same label can be redefined, and only the last definition will
  ever be accessed.

  As every extension of the $\lambda$-calculus, it defines an implicative
  structure $\A$ of sets of terms modulo $\beta\iota$-equivalence. Moreover, if
  $\ell$ is a label, we define a function $\Recl : \A \to \A$ by:
  \begin{align*}
    \Recl(a) = \left\{ t \mid [t.\ell]_{\beta\iota} \in a\right\}
  \end{align*}

  Intuitively $\Recl(a)$ is the set of terms that have a record of label $\ell$
  which contains an element of $a$.  We check that
  $\bigmeet_{i \in I} \Recl(a_i) = \left\{ t \mid \forall i \in I,
    [t.\ell]_{\beta\iota} \in a_i\right\} = \left\{ t \mid [t.\ell]_{\beta\iota}
    \in \cap_{i\in I} a_i\right\} = \Recl(\bigmeet_{i \in I} a_i)$.  Remark
  that for all $a \in \A$, $\Appl{\Recl(a)}=a$; this is a consequence of
  our choice to work in Kleene realizability (where equivalent terms are
  equated) and not in Krivine realizability (where an order relation abstracts
  the reduction).

  Moreover, given two distinct labels $\ell$ and $r$ and associated records
  $\Recl$ and $\Recr$, we have that:
  \begin{align*}
    \lefteqn{\Appl{(\Recl(a) \curlywedge \Recr(b))}}\\
    					&= \bigmeet \{ c \in \A \mid
                                             (\Recl(a) \curlywedge \Recr(b))
                                             \preq \Recl(c) \} \\
                                           &= \bigmeet \left\{ c \in \A \mid
                                             \left\{ t \mid
                                             [t.\ell]_{\beta\iota} \in a \wedge
                                             [t.r]_{\beta\iota} \in b \right\}
                                             \subseteq \left\{ t \mid
                                             [t.\ell]_{\beta\iota} \in c\right\}
                                             \right\}\\
                                           &= a = \Appl{(\Recl(a)}
  \end{align*}
  So $\Recl$ and $\Recr$ are compatible records.
\end{example}

\subsection{Additives}

\begin{definition}[additive connectives]
  Let $\Recl$ and $\Recr$ be two compatible records. We define:
  \begin{align*}
    a \with b &= \Recl(a) \curlywedge \Recr(b) \\
    a \oplus b &= \bigmeet_{c \in \A} \Recl(a \to c) \curlywedge \Recr(b
                 \to c) \to c \\
  \end{align*}
\end{definition}

\begin{proposition}
We have that: 
\[ a\oplus b \preq ((a\to \bot) \with (b \to \bot)) \to \bot.\]
\end{proposition}
\begin{proof}
  \begin{align*}
    a \oplus b &= \bigmeet_{c \in \A} \Recl(a \to c) \curlywedge \Recr(b
                 \to c) \to c \\
               &\preq \Recl(a \to \bot) \curlywedge \Recr(b \to \bot)
                 \to \bot \\
               &\preq ((a\to \bot) \with (b \to \bot)) \to \bot \qedhere
  \end{align*}
\end{proof}

\begin{definition}
  Let $(\A, \preq, \li)$ be an implicative structure with two
  compatible records $\Recl$ and $\Recr$.  An \emph{additive separator} is a
  separator $S$ on $\A$ such that
  \begin{enumerate}
  \item for all $a\in \A$, $a\to \Recl(a) \in S$ and $a\to \Recr(a) \in S$ ;
  \item for all $a\in \A$, $a\to \Appl{a} \in S$ and $a\to \Appr{a} \in S$;
  \item for all $a,b \in \A$, $(\Recl(a) \curlywedge \Recr(b)) \to (\Recl(b) 
  \curlywedge \Recr(a)) \in S$; \label{cond:records-commutativity}
  \item for all $a,b,c \in \A$,
  \begin{align*}
   (\Recl\Recl(a) \curlywedge \Recl\Recr(b) \curlywedge \Recr(c)) &\to (\Recl(a)
    \curlywedge \Recr\Recl(b) \curlywedge \Recr\Recr(c)) \in S,\\
    (\Recl(a)
    \curlywedge \Recr\Recl(b) \curlywedge \Recr\Recr(c)) &\to (\Recl\Recl(a)
    \curlywedge \Recl\Recr(b) \curlywedge \Recr(c)) \in S. 
    \end{align*}
    \label{cond:records-associativity}
  \end{enumerate}
  An \emph{additive implicative algebra} is a tuple
  $(\A, \preq, \li, \Recl,\Recr,S)$ where
  $(\A, \preq, \to)$ is an implicative structure, $\Recl$ and $\Recr$
  two compatible records and $S$ an additive separator.
\end{definition}

\begin{remark}
  Condition \ref{cond:records-commutativity} states the commutativity that a
  linear separator does not distinguish between different records: the labels
  are just names that allow to store and retrieve information in isomorphic
  containers.
  
  Condition \ref{cond:records-associativity} states the associativity of the
  nesting of records. Indeed, as records preserve arbitrary infima, the
  condition can be rewriten as
  \begin{enumerate}
  \item[\ref{cond:records-associativity}'] for all
    $a,b,c \in \A$, 
    \begin{align*}
    (\Recl(\Recl(a) \curlywedge \Recr(b)) \curlywedge \Recr(c)) &\to (\Recl(a)
    \curlywedge \Recr(\Recl(b) \curlywedge \Recr(c))) \in S,\\
    (\Recl(a) \curlywedge \Recr(\Recl(b) \curlywedge \Recr(c))) &\to
    (\Recl(\Recl(a) \curlywedge \Recr(b)) \curlywedge \Recr(c)) \in S.
    \end{align*}
  \end{enumerate}
\end{remark}

The interpretation of the $\lambda$-calculus can be extended to the
$\lambda$-calculus with records.


\begin{proposition}
  \label{prop:additives}
  Let $(\A, \preq, \li,S,\Recl,\Recr)$ be an additive implicative
  algebra. We have the following properties between additives:
  \begin{enumerate}
  \item $a\with b \equiprov_S b \with a$
  \item $a\oplus b \equiprov_S b \oplus a$
  \item $a \with b \vdash_S a$ and $a \with b \vdash_S b$
  \item $a \vdash_S a \oplus b$ and $b\vdash_S a \oplus
    b$ \label{prop:additives:oplus-colimit}
  \item $a\with \top \equiprov_S a \equiprov \top \with a$
  \item $a \oplus \bot \equiprov_S a \equiprov \bot \oplus a$
  \item $a\with (b \with c) \equiprov_S (a \with b) \with c$
  \item for all $d\in \A$ such that $a \vdash_S d$ and $b \vdash_S d$, $a
    \oplus b \vdash_S d$
  \item $a\oplus (b \oplus c) \equiprov_S (a \oplus b) \oplus c$
  \end{enumerate}
  as well as the distributivity laws of multiplicatives over additives:
  \begin{enumerate}
  \item $a \otimes (b \oplus c) \dashv_S (a \otimes b) \oplus (a \otimes c) $
  \item $a \multimap (b \with c) \vdash_S (a \multimap b) \with (a \multimap
    c)$
  \item $(a \oplus b) \multimap c \equiprov_S (a \multimap c) \with (b \multimap
    c)$ 
  \end{enumerate}
\end{proposition}

\subsection{Additives in linear realisability models}

The above definition of additives relates closely to the definition of additive connectives in some of the models from the literature. In particular, we will focus on Girard's model \citep{goi3}, and more precisely the operator-algebraic formulation \citep{Duchesne-phd}\citep[Section 4.1]{seiller-phd} in which additives are treated using partial isometries $p, q \in \mathcal{L}(\mathbb{H})$, where $\mathbb{H}$ is a Hilbert space with countable basis, say $\ell^2(\mathbb{N})$. These partial isometries may be defined as follows on the natural basis of $\mathbb{H}$: $p(b_i) = b_{2i}$, and $q(b_i)=b_{2i+1}$. They have conjugates $p^*, q^*$, and they satisfy the following properties: 
\[ p^* p = 1 = q^* q; \hspace{2em} q^*p = 0 = p^* q; \hspace{2em} pp^* + qq^* = 1, \]
where $1$ denotes the identity operator $x\mapsto x$ and $0$ denotes the zero operator $x\mapsto 0$, both in $\mathcal{L}(\mathbb{H})$.

We can recover the operations $\Appl{a}$ and $\Recl(a)$ as follows: 
\begin{align*}
\Appl{a} &= (1\otimes p^*)a\otimes 1(1\otimes p) \\
\Recl(a) &= (1\otimes p)a\otimes 1(1\otimes p^*)
\end{align*}
It is then easy to check that the interpretation of additives follows the constructions given in the preceding section. Moreover, the needed properties of additive separators are satisfied. In particular, we notice that the existence of the operators $s, t$, used by Girard in his constructions\footnote{From the first geometry of interaction model \citep{goi1} which, even though did not interpret additive connectives, introduced the operators $p, q, s, t$ considered here.} \citep{goi1}, correspond to the last two properties. Indeed, those operators satisfy:
\begin{align*}
s s^* = s^* s = 1 \quad s (u\otimes v) = (v\otimes u)s,\\
tt^* = t^*t =1 \quad t(u\otimes (v\otimes w)) = ((u\otimes v)\otimes w)t,
\end{align*}
where $u\otimes v$ is a notation for $p^* u p+q^* v q$.


\subsection{Fixpoint exponential}

An implicative structure is a complete lattice; it is in particular possible to
use the lattice structure itself to define some operations, such as the
exponentials. In particular, Baelde \citep{Baelde12} has introduced, in the
framework of \(\mu\)MALL, that is, multiplicative-additive linear logic extended
with fixed points, an encoding of the exponentials as fixpoints. We will follow
this encoding.

Let us first remark that, in any additive implicative algebra, and \(X\in \A\),
the applications:
\begin{align*}
  a \mapsto \1 \with X \with (a \otimes a)\\
  a \mapsto \bot \oplus X \oplus (a \parr a)
\end{align*}
are monotonic. Then, as \(A\) is a complete lattice, by Knaster-Tarski theorem,
they both have complete lattices of fixpoints. We can then define the
exponentials to be any of these fixpoints. A canonical choice can the be:
\begin{align*}
  \oc X &:= \nu a. \1 \with X \with (a \otimes a)\\
  \wn X &:= \mu a. \bot \oplus X \oplus (a \parr a)
\end{align*}
where \(\mu\) denotes the least fixpoint and \(\nu\) the greatest fixpoint. It
is fairly easy to check that \(\oc\) verifies the axioms of exponential we gave
in Section \ref{sec:exponentials}. Precisely:

\begin{proposition}
  Let $(\A, \preq, \li, \Recl, \Recr, S)$ be an additive implicative algebra.
  
  The \(\oc:\A \to \A\) function defined as the greatest fixpoint \(\oc X\)of
  \(a \mapsto \1 \with X \with (a \otimes a)\) is monotonic and \(S\) is an
  exponential separator for this function: $(\A, \preq, \li, \oc, S)$ is an
  exponential implicative algebra.
\end{proposition}
\section{Conclusion}

We have shown that linear realizability, as defined by Seiller merging different
strands of work originating in Girard's investigations of linear logic can be
unified with forcing, intutionnistic and classical realizability in the
framework of implicative algebras, embodying a general intuition that all these
works are connected.

This unification opens a lot of research directions, in particular:
\begin{itemize}
\item Miquey \citep{miquey} defined disjunctive and conjunctive algebras so as
  to study call-by-value and call-by-name \(\lambda\)-calculus. Linear logic
  being another tool used to relate these calculi, we expect it to be related;
\item implicative algebras have been used to define topoi (and this construction
  is actually the core of forcing). We can expect such topoi to be decomposable
  through the linear decomposition of the implicative algebra;
\item many constructions build models of classical linear logic (where negation
  is involutive) from models of intutionnistic linear logic through a kind of
  Chu construction \citep{shulman}. This can also be investigated for linear
  implicative algebras.
\end{itemize}

\bibliographystyle{abbrv}
\bibliography{biblio}

\appendix
\clearpage
\section{Omitted proofs}

\begin{proof}[Proof of Proposition \ref{prop:applicative}]
    For all $a, b \in \A$, we write $V_{a,b} = \{c \in \A : c \cdot a \preq b\}$, so that $a \leadsto b := \bigjoin V_{a,b}$.
    \begin{enumerate}
        \item let $a' \preq a$, $b \preq b'$, then $\forall c \in \A : c \cdot a' \preq c \cdot a$ and $b \preq b'$ so $V_{a,b} \subseteq V_{a',b'}$ and then $a \leadsto b = \bigjoin V_{a,b} \preq \bigjoin V_{a',b'} = a' \leadsto b'$.
    
        \item It is clear that $a \in V_{b,a \cdot b}$, hence $a \preq \bigjoin V_{b,a \cdot b} = a \preq (b \leadsto a \cdot b)$.
    
        \item We have $(a \leadsto b) \cdot a = (\bigjoin V_{a,b}) \cdot a  = \bigjoin_{c \in V_{a,b}} (c \cdot a) \preq b$, from the definition of $V_{a,b}$.
    
        \item From (3), $a \leadsto b \in V_{a,b}$, so $a \leadsto b = \min(V_{a,b})$
    
        \item Assuming that $a \cdot b \preq c$, we have $a \in V_{b,c}$ so $a \preq b \leadsto c$. Conversely, if $a \preq (b \leadsto c)$, from (3) we have $a \cdot b \preq (b \leadsto c) \cdot b \preq c$. 
    \end{enumerate}
\end{proof}

\begin{proof}[Proof of Proposition \ref{prop:applicativevsimplicative}]
     we already prooved that $\leadsto$ verify $(1)$.

    We have 
    \begin{align*}
        \forall b \in B, \bigmeet_{b \in B} b \preq b 
        & \implies \forall b \in B, (a \leadsto \bigmeet_{b \in B} b) \preq a \leadsto b \\
        & \implies (a \leadsto \bigmeet_{b \in B} b) \preq \bigmeet_{b \in B} (a \leadsto b)
    \end{align*}
    and forall $z \in \A$
    \begin{align*}
        z \preq \bigmeet_{b \in B} (a \leadsto b) 
        & \implies \forall b \in B, z \preq (a \leadsto b) \\ 
        & \implies \forall b \in B, z \cdot a \preq b \\
        & \implies z \cdot a \preq \bigmeet_{b \in B} b \\
        & \implies z \preq a \leadsto \bigmeet_{b \in B} b
    \end{align*}
    so 
    $$
        \bigmeet_{b \in B} (a \leadsto b) \preq a \leadsto \bigmeet_{b \in B} b
    $$
    then $\leadsto$ verify $(2)$.

    And finally for all $a,b \in \A$ :
    \begin{align*}
        ab & = \bigmeet \{c\in \A : a \preq b \leadsto c\} \\ 
        & = \bigmeet \{c\in \A : a \preq \max \{z \in \A : z \cdot b \preq c\}\} \\
        & = \bigmeet \{c\in \A : a \cdot b \preq c\} \\
        & = a \cdot b
    \end{align*}
\end{proof}

\begin{proof}[Proof of Proposition \ref{prop:implicativetoapplicative}]
    We already know that the application respects 1.
    
    For 2. we have that, for all $b \in \A$, $A \subseteq \A$ :
    \begin{align*}
        \left(\bigjoin_{a \in A} a\right)b & = \bigmeet \{c\in \A : \bigjoin_{a \in A} a \preq b \to c\} \\ 
        & = \bigmeet \{c\in \A : \forall a \in A, a \preq b \to c\} \\
        & = \bigmeet \{c\in \A : \forall a \in A, ab \preq c\} \\
        & = \bigmeet \{c\in \A : \bigjoin_{a \in A} (ab) \preq c\} \\
        & = \bigjoin_{a \in A} (ab)
    \end{align*}

    Lastly we have for all $a,b \in \A$ :
    \begin{align*}
        a \leadsto b & = \bigjoin \{c\in \A : ca \preq b\} \\ 
        & = \bigjoin \{c\in \A : \min \{z \in \A : c \preq a \to z\} \preq b\} \\
        & = \bigjoin \{c\in \A : c \preq a \to b\} \\
        & = a \to b \\
    \end{align*}
\end{proof}

\begin{proof}[Proof of Proposition \ref{prop:semantictyping}]
    Axiom, Parameter, Subsumption : immediate with the definition.

    Context subsumption : follow by \ref{monotony of substitution} (monotony of substitution).

    $\li$-R : Let assume that $\fv(t) \subseteq \dom(\Gamma, x : a)$ and $(t[\Gamma, x := a])^\A \preq b$. we have that $\fv(\lambda x.t) \subseteq \dom(\Gamma)$ and $x \not\in \dom(\Gamma)$, so that :
    \begin{align*}
        ((\lambda x.t)[\Gamma])^\A
        & = (\lambda x.t[\Gamma])^\A \\
        & = \bigmeet_{a_0 \in \A} \left(a_0 \li (t[\Gamma, x := a_0])^\A \right) \\
        & \preq a \li (t[\Gamma, x := a])^\A \\
        & \preq a \li b
    \end{align*}

    Cut : Let assume that $\fv(t) \subseteq \dom(\Gamma)$, $(t[\Gamma])^\A \preq a$, $\fv(u) \subseteq \dom(\Delta, x : a)$ and $(u[\Delta, x := a])^\A \preq b$. We have that $\fv(\lambda x.u) \subseteq \dom(\Delta)$ and $x \not\in \dom(\Delta)$. We can suppose that $\dom(\Gamma) \cap \dom(\Delta) = \emptyset$ and $x \not\in \dom(\Gamma)$, so that :
    \begin{align*}
        (((\lambda x.u)t)[\Gamma^, \Delta])^\A
        & = (\lambda x.u[\Delta])^\A (t[\Gamma])^\A \\
        & \preq \bigmeet_{a_0 \in \A} \left(a_0 \li (u[\Delta, x := a_0])^\A \right) a \\
        & \preq (a \li (u[\Delta, x := a])^\A) a \\
        & \preq (a \li b)a \\
        & \preq b
    \end{align*}

    $\li$-L : Let assume that $\fv(t) \subseteq \dom(\Gamma)$, $(t[\Gamma])^\A \preq a$, $\fv(u) \subseteq \dom(\Delta, x : b)$ and $(u[\Delta, x := b])^\A \preq c$. We have that $\fv(\lambda x.u) \subseteq \dom(\Delta)$ and $x \not\in \dom(\Delta)$. We can suppose that $\dom(\Gamma) \cap \dom(\Delta) = \emptyset$, $x \not\in \dom(\Gamma)$, $y \not\in \dom(\Gamma)$ and $y \not\in \dom(\Delta)$, so that :
    \begin{align*}
        \lefteqn{(((\lambda x.u)(yt))[\Gamma^, \Delta, y := a \li b])^\A}\\
        & = (\lambda x.u[\Delta])^\A ((y[y := a \li b])^\A(t[\Gamma])^\A) \\
        & = \bigmeet_{b_0 \in \A} \left(b_0 \li (u[\Delta, x := b_0])^\A \right) ((a \li b) a) \\
        & \preq (b \li (u[\Delta, x := b])^\A) b \\
        & \preq (b \li c)b \\
        & \preq c
    \end{align*}

    Generalisation : $(t[\Gamma])^\A \preq a_i$ for all $i \in I$ implies that $(t[\Gamma])^\A \preq \bigmeet_{i \in I} a_i$.
\end{proof}

\begin{proof}[Proof of proposition \ref{prop:permutations}]
    We have :
    \begin{align*}
        \lambda_\sigma t t_1 ... t_n
            & = \lambda x y_1 \ldots y_m.x y_{\sigma^{-1}(1)} \ldots y_{\sigma^{-1}(m)} t_1 ... t_n \\
            & \bred t t_{\sigma^{-1}(1)} \ldots t_{\sigma^{-1}(m)} t_{m+1} ... t_n 
    \end{align*}
    and
    \begin{align*}
        \lambda_\tau t t_1 ... t_n
            & = \lambda x y_1 \ldots y_n.x y_{\tau^{-1}(1)} \ldots y_{\tau^{-1}(n)} t_1 ... t_n \\
            & \bred t t_{\tau^{-1}(1)} \ldots t_{\tau^{-1}(n)} \\
            & = t t_{\sigma^{-1}(1)} \ldots t_{\sigma^{-1}(m)} t_{m+1} ... t_n 
    \end{align*}
\end{proof}

\begin{proof}[Proof of Lemma \ref{lem:linearcombiterms}]
    We procede by induction on the structure of $t$ :

    If $t$ is a free variable, we take $\I$.

    If $t = uv$, then for all $x_1, \ldots, x_n \in \fv(t)$, $i \leq n$, either $x_i \in \fv(u)$ either $x_i \in \fv(v)$.
    By \ref{combinatory permutations}\footnote{Take a linear combinatory term $T \bred \lambda_\sigma$ and then use $(T\cmb_{k,n})$ instead of $\cmb_{k,n}$} we can suppose that there is a $k \leq n$ such that $x_1, \ldots, x_k \in \fv(u)$ and $x_{k+1}, \ldots, x_n \in \fv(v)$.
    By inductions, there are two linear combinatory terms $u_0 \bred \lambda x_1 \ldots x_k.u$ and $v_0 \bred \lambda x_{k+1} \ldots x_n.v$. Then we have
    \begin{align*}
        \cmb_{k,n}u_0v_0 
        & \bred \lambda x_1 \ldots x_n.u_0 x_1 \ldots x_k (v_0 x_{k+1} \ldots x_n) \\
        & \bred \lambda x_1 \ldots x_n.uv \\
        & = \lambda x_1 \ldots x_n.t
    \end{align*}

    If $t = \lambda y.u$ then for all $x_1, \ldots, x_n \in \fv(t)$, $x_1, \ldots, x_n, y \in \fv(u)$, so by induction there is a linear combinatory term $u_0 \bred \lambda x_1 \ldots x_n y.u = \lambda x_1 \ldots x_n.t$. 
    
\end{proof}

\begin{proof}[Proof of Proposition \ref{prop:lambdaclosure}]
    Let $t$ be a linear $\lambda$-term with free variables $\overline{x}$ and let $\overline{a} \in S$ be parameters in $S$. By the theorem \ref{linear combinatory} there is a closed linear combinatory term $t_0$ such that $t_0 \bred \lambda\overline{x}.t$. from the properties $(2)$ and (3') of separators, $t_0^\A \overline{a} \in S$. Moreover, from \ref{reductions},
    $$
        t_0^\A\overline{a} \quad \preq \quad (\lambda\overline{x}.t)^\A \overline{a} \quad \preq \quad (t[\overline{x} := \overline{a}])^\A
    $$
    then $(t[\overline{x} := \overline{a}])^\A \in S$ by upward closure.
\end{proof}

\begin{proof}[Proof of Proposition \ref{prop:semantictyping:tensor}]
    $\tens$-R : With the given premises we have :
    \begin{align*}
        ((\lambda z.ztu)[\Gamma, \Delta])^\A
        & = (\lambda z.zt[\Gamma]u[\Delta])^\A \\
        & = \bigmeet_{d \in \A} \left(d \li (z[z := d])^\A(t[\Gamma])^\A(u[\Delta])^\A \right) \\
        & \preq \bigmeet_{d \in \A} \left(d \li dab \right) \\
        & \preq (a \li b \li c) \li ((a \li b \li c)ab) \\
        & \preq (a \li b \li c) \li c
    \end{align*}
    
    $\tens$-L : With the given premise we have :
    \begin{align*}
        \lefteqn{((z(\lambda xy.t))[\Gamma, z := a \tens b])^\A}\\
        & = (z[z := a \tens b](\lambda xy.t)[\Gamma])^\A \\
        & = (a \tens b) \bigmeet_{a_0, b_0 \in \A} \left(a_0 \li  b_0 \li (t[\Gamma, x = a_0, y = b_0])^\A\right) \\
        & \preq \bigmeet_{c_0 \in \A} ((a \li b \li c_0) \li c_0) \bigmeet_{a_0, b_0 \in \A} (a_0 \li b_0 \li c) \\
        & \preq ((a \li b \li c) \li c) (a \li b \li c) \\
        & \preq c
    \end{align*}
\end{proof}

\begin{proof}[Proof of \ref{prop:MLLsoundness}]
    By induction, we check that if a typing jugement $\Gamma \vdash t : a$ is derived using the semantic typing rules (Axiom), (Cut), ($\li$-R), ($\li$-L), ($\tens$-R) and ($\tens$-L) from \ref{semantic typing} and \ref{semantic typing tensor}, then $\dom(\Gamma) = \fv(t)$ and $t$ is linear.
    So, by induction on the derivation of $A$, we use those rules to construct a linear $\lambda$-term $t$ such that $\vdash t : A^\A$. 
    Then we have $t^\A \preq A^\A$ and we conclude by \ref{lambda closure} (linear $\lambda$-closure).
\end{proof}

\begin{proof}[Proof of \ref{prop:quotientdef}]
    \begin{enumerate}
        \item See \citep{DBLP:journals/mscs/Miquel20}
        \item If $a \ent_S a'$ and $b \ent_S b'$, then $T = \lambda x.x(\lambda yzt.t((a \li a')y)((b \li b')z))$ is a linear $\lambda$-term with parameters in $S$ so $T^\A \in S$.
        We easily check that $T^\A \preq a \tens b \li a' \tens b'$ so $a \tens b \ent_S a' \tens b'$.
    \end{enumerate}
\end{proof}

\begin{proof}[Proof of Proposition \ref{prop:quotientproperties}]
    for each point we construct a linear $\lambda$-term with the required type :
    \begin{enumerate}
        \item 
        $\lambda x.x(\lambda aby.yba) : \cl a \tens \cl b \li \cl b \tens \cl a$
        
        \item 
        $\lambda x.x(\lambda yc.y(\lambda abt.ta(\lambda u.ubc))) : (\cl a \tens \cl b) \tens \cl c \li \cl a \tens (\cl b \tens \cl c)$

        $\lambda x.x(\lambda ay.y(\lambda bct.t(\lambda u.uab)c)) : \cl a \tens (\cl b \tens \cl c) \li (\cl a \tens \cl b) \tens \cl c$
        
        \item 
        $\lambda x.x(\lambda ai.ia) : \cl a \tens \cl{\I^\A} \li \cl a$

        $\lambda a. (\lambda x.xa\I) : \cl a \li \cl a \tens \cl{\I^\A}$
        
        \item 
        $\lambda xab.x(\lambda y.yab) : (\cl a \tens \cl b \li \cl c) \li \cl a \li \cl b \li \cl c$

        $\lambda xy.y (\lambda ab.xab) : (\cl a \li \cl b \li \cl c) \li \cl a \tens \cl b \li \cl c$
        
        \item 
        $\lambda x.x(\lambda ya.ya) : (\cl a \li \cl b) \tens \cl a \li \cl b$
        
        \item 
        $\lambda abx.xab : \cl a \li \cl b \li \cl a \tens \cl b$
        
        \item 
        $\lambda xa.x(\lambda yz.z(ya)) : (\cl a \li \cl b) \tens (\cl b \li \cl c) \li \cl a \li \cl c$
        
        \item 
        $\lambda xy.x(\lambda uv.y(\lambda act.t(ua)(vc))) : (\cl a \li \cl b) \tens (\cl c \li \cl d) \li \cl a \tens \cl c \li \cl b \tens \cl d$
    \end{enumerate}
\end{proof}

\begin{proof}[Proof of Proposition \ref{prop:linreallattice}]
    By definition $(\bigcup_{i \in I} \mathbf{A}_i)^{\simperp\simperp}$ is the smallest type containing all the $\mathbf{A}_i$ so it is the join.

    By definition $\bigmeet_{i \in I} \mathbf{A}_i \subseteq \bigcap_{i \in I} \mathbf{A}_i$ and we have for all $T_i \subseteq \Proj, i \in I$ :
    \begin{align*}
        \bigcap_{i \in I} \sorth{T_i}
        & = \bigcap_{i \in I} \{\proj{a} \in \Proj : \forall \proj{p} \in T_i, \proj{a} \simperp_{\antipode} \proj{p}\} \\
        & = \{\proj{a} \in \Proj : \forall \proj{p} \in \bigcap_{i \in I} T_i, \proj{a} \simperp_{\antipode} \proj{p}\} \\
        & = \sorth{(\bigcap_{i \in I} T_i)}
    \end{align*}
    So in particular, $\types$ is closed by intersection, and then $\bigmeet_{i \in I} \mathbf{A}_i = \bigcap_{i \in I} \mathbf{A}_i$
\end{proof}

\begin{proof}[Proof of Proposition \ref{prop:additives}]
  \begin{enumerate}
  \item This is exactly condition \ref{cond:records-commutativity};
  \item $a\oplus b \equiprov_S b \oplus a$
  \item We have that $(a \with b) \to \Appl{(a \with b)} \in S$, and moreover,
    $\Appl{(a \with b)} = \Appl{(\Recl(a) \curlywedge \Recr(b))} =
    \Appl{\Recl(a)} \preq a$, so $(a \with b) \to a \in S$;
  \item 
  \item 
    $a \with \top = \Recl(a) \curlywedge \Recr(\top) = \Recl(a) \curlywedge \top
    = \Recl(a)$, as $\Recr$ preserve the nullary infimum. So, $a \to (a\with
    \top) \in S$. The converse direction is a consequence of a previous bullet.
  \item $a \oplus \bot \equiprov_S a \equiprov \bot \oplus a$
  \item This is exactly condition \ref{cond:records-associativity};
    $a\with (b \with c) \equiprov_S (a \with b) \with c$
  \item $\lambda t. t \Recl(a\to d) \Recr(b\to d) \preq (a\oplus b) \to d \in
    S$
  \item by the previous item, as $a \vdash_S a \oplus (b \oplus c)$,
    $b \vdash_S a \oplus (b \oplus c)$, and $c \vdash_S a \oplus (b \oplus c)$,
    we have that $(a \oplus b) \vdash_S a \oplus (b \oplus c)$ and,
    consequently, $(a \oplus b) \oplus c \vdash_S a \oplus (b \oplus c)$. The
    other direction is similar.
  \end{enumerate}

  For the second part of the proposition.

  For the first statement, we can check that
    \[ \lambda p. p(\lambda xy.\lambda z. zx(\lambda v. (\Appl{v}) u)) \preq a
    \otimes b \to a \otimes (b \oplus c) \in S. \] 
    Similarly,
    \[ \lambda p. p(\lambda xy.\lambda z. zx(\lambda v. (\Appr{v}) u)) \preq a
    \otimes c \to a \otimes (b \oplus c) \in S.\] 
    So, by \Cref{prop:additives}.\ref{prop:additives:oplus-colimit}, 
    \[ (a \otimes b) \oplus (a \otimes c) \vdash_S a \otimes (b \oplus c). \]
    
    \begin{align*}
      &a \otimes (b \oplus c)\\
      =& \bigmeet_{d \in \A} a \to \left(
                               \bigmeet_{e \in \A} \Recl(b \to e)
                               \curlywedge \Recr(c \to e) \to e \right) \to d\\
      &(a \otimes b) \oplus (a \otimes c) \\=& \bigmeet_{d \in \A} \left(
                                           \Recl\left( \left( \bigmeet_{e \in \A}
                                           a \to b \to e \right) \to d \right)
                                           \curlywedge  \Recr\left( \left(
                                           \bigmeet_{f \in \A} a \to c \to
                                           f \right) \to d \right) \right) \to d
    \end{align*}    
     $\lambda z. z(\lambda x. \Appl{(a \multimap (b \with c)x)}) (\lambda
    x. \Appr{(a \multimap (b \with c)x)}) \preq $
    $a\multimap (b \with c) = a \to (\Recl(b) \curlywedge \Recr(c)) = (a
    \to \Recl(b)) \curlywedge (a \to \Recr(c)) \curlywedge (a \multimap b) \with
    (a \multimap c)$
\end{proof}

\end{document}